\documentclass[pra,superscriptaddress,twocolumn,showkeys,letterpaper]{revtex4}%
\usepackage{amsfonts}
\usepackage{amsmath}
\usepackage{amssymb}
\usepackage{graphicx}
\usepackage{hyperref}
\usepackage{color}%
\providecommand{\U}[1]{\protect\rule{.1in}{.1in}}
\newtheorem{theorem}{Theorem}

\newtheorem{corollary}[theorem]{Corollary}

\newtheorem{remark}[theorem]{Remark}

\newenvironment{proof}[1][Proof]{\noindent\textbf{#1.} }{\ \rule{0.5em}{0.5em}}

\def\>{\rangle}
\def\<{\langle}

\begin{document}
\preprint{ }
\title[Approximate reversibility]{Approximate reversibility in the context of entropy gain, information gain,
and complete positivity}
\author{Francesco Buscemi}
\affiliation{Department of Computer Science and Mathematical Informatics, Nagoya
University, Chikusa-ku, Nagoya, 464-8601, Japan}
\email{buscemi@is.nagoya-u.ac.jp}
\author{Siddhartha Das}
\affiliation{Hearne Institute for Theoretical Physics, Department of Physics and Astronomy,
Louisiana State University, Baton Rouge, Louisiana 70803, USA}
\email{sdas21@lsu.edu}
\author{Mark M. Wilde}
\affiliation{Hearne Institute for Theoretical Physics, Department of Physics and Astronomy,
Louisiana State University, Baton Rouge, Louisiana 70803, USA}
\affiliation{Center for Computation and Technology, Louisiana State University, Baton
Rouge, Louisiana 70803, USA}
\email{mwilde@lsu.edu}

\begin{abstract}
There are several inequalities in physics which limit how well we can process
physical systems to achieve some intended goal, including the second law of
thermodynamics, entropy bounds in quantum information theory, and the
uncertainty principle of quantum mechanics. Recent results provide physically
meaningful enhancements of these limiting statements, determining how well one
can attempt to reverse an irreversible process. In this paper, we apply and
extend these results to give strong enhancements to several entropy
inequalities, having to do with entropy gain, information gain, entropic
disturbance, and complete positivity of open quantum systems dynamics. Our
first result is a remainder term for the entropy gain of a quantum channel.
This result implies that a small increase in entropy under the action of a
subunital channel is a witness to the fact that the channel's adjoint can be
used as a recovery map to undo the action of the original channel. We
apply this result to pure-loss, quantum-limited amplifier, and
phase-insensitive quantum Gaussian channels, showing how a quantum-limited
amplifier can serve as a recovery from a pure-loss channel and vice versa. Our
second result regards the information gain of a quantum measurement, both
without and with quantum side information. We find here that a small
information gain implies that it is possible to undo the action of the
original measurement if it is efficient. The result also has operational
ramifications for the information-theoretic tasks known as measurement
compression without and with quantum side information. Our third result shows
that the loss of Holevo information caused by the action of a noisy channel on
an input ensemble of quantum states is small if and only if the noise can be
approximately corrected on average. We finally
establish that the reduced dynamics of a system-environment interaction are
approximately completely positive and trace-preserving if and only if the data
processing inequality holds approximately.

\end{abstract}
\volumeyear{ }
\volumenumber{ }
\issuenumber{ }
\eid{ }
\date{\today}
\startpage{1}
\endpage{10}
\maketitle

\section{Introduction}

The second law of thermodynamics constitutes a fundamental limitation on our
ability to extract energy from physical systems \cite{C85,BP02,Y01}. The data
processing inequality represents a limitation on our ability to process
information, being the basis for most of the important capacity theorems in
quantum information theory \cite{W15book}. The entropic uncertainty principle
of quantum mechanics places a limitation on how well we can measure
incompatible observables \cite{1925heisenberg,CBTW15}. These seemingly
disparate statements have a common mathematical foundation in an entropy
inequality known as the monotonicity of quantum relative entropy
\cite{Lin75,Uhl77}, which states that the quantum relative entropy cannot
increase under the action of a quantum channel. More precisely, the quantum
relative entropy between two density operators $\rho$ and $\sigma$ is defined
as \cite{U62}%
\begin{equation}
D(\rho\Vert\sigma)\equiv\operatorname{Tr}\{\rho\left[  \log\rho-\log
\sigma\right]  \},
\end{equation}
and the monotonicity of quantum relative entropy states that
\cite{Lin75,Uhl77}%
\begin{equation}
D(\rho\Vert\sigma)\geq D(\mathcal{N}(\rho)\Vert\mathcal{N}(\sigma)),
\label{eq:rel-ent-mono}%
\end{equation}
where $\mathcal{N}$ is a quantum channel.

Recently, researchers have explored refinements of these statements in various
contexts, with the common theme being to understand how well one can attempt
to reverse an irreversible process. One of the main technical developments
which has allowed for these refined statements is a strengthening of the
monotonicity of quantum relative entropy of the following form \cite{W15}:%
\begin{equation}
D(\rho\Vert\sigma)\geq D(\mathcal{N}(\rho)\Vert\mathcal{N}(\sigma))-\log
F(\rho,(\mathcal{R}\circ\mathcal{N})(\rho)), \label{eq:mono-enhance}%
\end{equation}
where $F(\omega,\tau)\equiv\left\Vert \sqrt{\omega}\sqrt{\tau}\right\Vert
_{1}^{2}$ is the quantum fidelity \cite{U73} between two density operators
$\omega$ and $\tau$, and $\mathcal{R}$ is a recovery channel with the property
that it perfectly recovers the $\sigma$ state, in the sense that
$\sigma=(\mathcal{R}\circ\mathcal{N})(\sigma)$ (see also
\cite{Sutter15,Junge15} for later developments and \cite{FR14}\ for an
important earlier development with conditional mutual information).

Several applications follow as a consequence. Ref.~\cite{WWW15} gave an
application in thermodynamics, proving that if the free energies of two states
are close and if it is possible to transition from one state to another via a
thermal operation such that there is an energy gain in the process, then one
can approximately reverse this thermodynamic transition without using any
energy at all. Ref.~\cite{BWW15} showed how to tighten the uncertainty
principle in the presence of quantum memory \cite{BCCRR10}\ with another term
related to how much disturbance a given measurement causes, thus unifying
several aspects of quantum physics, including measurement incompatibility,
entanglement, and measurement disturbance, in a single entropic uncertainty
relation. Finally, Ref.~\cite{W15} has given an increased understanding of
many well known entropy inequalities in quantum information, such as the joint
convexity of quantum relative entropy, the non-negativity of quantum discord,
the Holevo bound, and multipartite information inequalities.

In this paper, we continue with this theme and derive several new results:

\begin{enumerate}
\item First, we give a strong improvement of the well known statement that the
quantum entropy cannot decrease under the action of a unital quantum channel
(a channel which preserves the identity operator). The bound that we derive
has a rather simple proof, following from the operator concavity of the
logarithm (related to the method used in \cite{Sutter15}). The main physical
implication of this result is that if the entropy gain under the action of a
unital channel is not too large, then it is possible to reverse the action of
this channel by applying its adjoint (which is a quantum channel in this case).

\item Next, we consider the information gain of a quantum measurement, a
concept introduced in \cite{G71} and subsequently refined in
\cite{Win04,BHH08,WHBH12,BRW14}. The information gain of a quantum measurement
quantifies how much data we can gather by performing a quantum measurement on
a given state. It has an operational interpretation as the rate at which a
sender needs to transmit classical information to a receiver in order for them
to simulate a quantum measurement on a given state \cite{Win04}. Here, we
prove that if the information gain is not too large, then it is possible to
reverse the action of the measurement and, in the operational context, one can
also simulate the measurement well on average without sending any classical
data at all. The result also applies if the measurement is performed on one
share of a bipartite state.

\item Third, we provide a clear operational meaning for the notion of
\textit{entropic disturbance}, defined in~\cite{BH08} as the loss of the
Holevo information due to the action of a noisy channel on an initial ensemble
of quantum states. We accomplish this by showing that a small loss of Holevo
information implies that the action of the noisy channel on the input ensemble
can be approximately undone, on average. This result answers a question left
open from \cite{BH08}.

\item Finally, we give a refinement of the recent link between the data
processing inequality and complete positivity of open quantum systems dynamics
\cite{B14}. In \cite{B14}, it was shown that the data processing inequality
holds if and only if the reduced dynamics of an evolution can be described by
a completely positive trace-preserving map. Here, we show how this result
holds approximately, which should allow for experimental tests if desired.
That is, we show that the data processing inequality holds approximately if
and only if the reduced dynamics of an evolution are approximately completely
positive and trace-preserving (see Section~\ref{sec:CP-DP}\ for precise statements).
\end{enumerate}

The rest of the paper is devoted to giving more details and explanations of
these results. We begin in the next section by setting notation, definitions,
and reviewing the prior literature in more detail. We then follow with each of
the aforementioned results and conclude in Section~\ref{sec:conclusion}\ with
a summary.

\section{Preliminaries}

\label{sec:notation}This section reviews background material on quantum
information, all of which is available in \cite{W15book}. Let $\mathcal{L}%
(\mathcal{H})$ denote the algebra of bounded linear operators acting on a
Hilbert space $\mathcal{H}$. Let $\mathcal{L}_{+}(\mathcal{H})$ denote the
subset of positive semi-definite operators. We also write $X\geq0$ if
$X\in\mathcal{L}_{+}(\mathcal{H})$. An\ operator $\rho$ is in the set
$\mathcal{D}(\mathcal{H})$\ of density operators (or states) if $\rho
\in\mathcal{L}_{+}(\mathcal{H})$ and Tr$\left\{  \rho\right\}  =1$. The tensor
product of two Hilbert spaces $\mathcal{H}_{A}$ and $\mathcal{H}_{B}$ is
denoted by $\mathcal{H}_{A}\otimes\mathcal{H}_{B}$ or $\mathcal{H}_{AB}%
$.\ Given a multipartite density operator $\rho_{AB}\in\mathcal{D}%
(\mathcal{H}_{A}\otimes\mathcal{H}_{B})$, we unambiguously write $\rho_{A}%
=\ $Tr$_{B}\left\{  \rho_{AB}\right\}  $ for the reduced density operator on
system $A$. We use $\rho_{AB}$, $\sigma_{AB}$, $\tau_{AB}$, $\omega_{AB}$,
etc.~to denote general density operators in $\mathcal{D}(\mathcal{H}%
_{A}\otimes\mathcal{H}_{B})$, while $\psi_{AB}$, $\varphi_{AB}$, $\phi_{AB}$,
etc.~denote rank-one density operators (pure states) in $\mathcal{D}%
(\mathcal{H}_{A}\otimes\mathcal{H}_{B})$ (with it implicit, clear from the
context, and the above convention implying that $\psi_{A}$, $\varphi_{A}$,
$\phi_{A}$ may be mixed if $\psi_{AB}$, $\varphi_{AB}$, $\phi_{AB}$ are pure).
A purification $|\phi^{\rho}\rangle_{RA}\in\mathcal{H}_{R}\otimes
\mathcal{H}_{A}$ of a state $\rho_{A}\in\mathcal{D}(\mathcal{H}_{A})$ is such
that $\rho_{A}=\operatorname{Tr}_{R}\{|\phi^{\rho}\rangle\langle\phi^{\rho
}|_{RA}\}$. An isometry $U:\mathcal{H}\rightarrow\mathcal{H}^{\prime}$ is a
linear map such that $U^{\dag}U=I_{\mathcal{H}}$. Often, an identity operator
is implicit if we do not write it explicitly (and should be clear from the context).

Throughout this paper, we take the usual convention that $f(A)=\sum_{i:a_{i}
\neq0}f(a_{i})|i\rangle\langle i|$ when given a function $f$ and a Hermitian
operator $A$ with spectral decomposition $A=\sum_{i}a_{i}|i\rangle\langle i|$.
In particular, $A^{-1}$ is interpreted as a generalized inverse, so that
$A^{-1}=\sum_{i:a_{i}\neq0}a_{i}^{-1}|i\rangle\langle i|$, $\log\left(
A\right)  =\sum_{i:a_{i}>0}\log\left(  a_{i}\right)  |i\rangle\langle i|$,
$\exp\left(  A\right)  =\sum_{i:a_{i}\neq0}\exp\left(  a_{i}\right)
|i\rangle\langle i|$, etc. Throughout the paper, we interpret $\log$ as the
binary logarithm. We employ the shorthand supp$(A)$ and ker$(A)$ to refer to
the support and kernel of an operator $A$, respectively.

A linear map $\mathcal{N}_{A\rightarrow B}:\mathcal{L}(\mathcal{H}%
_{A})\rightarrow\mathcal{L}(\mathcal{H}_{B})$\ is positive if $\mathcal{N}%
_{A\rightarrow B}\left(  \sigma_{A}\right)  \in\mathcal{L}(\mathcal{H}%
_{B})_{+}$ whenever $\sigma_{A}\in\mathcal{L}(\mathcal{H}_{A})_{+}$. Let
id$_{A}$ denote the identity map acting on a system $A$. A linear map
$\mathcal{N}_{A\rightarrow B}$ is completely positive if the map
id$_{R}\otimes\mathcal{N}_{A\rightarrow B}$ is positive for a reference system
$R$ of arbitrary size. A linear map $\mathcal{N}_{A\rightarrow B}$ is
trace-preserving if $\operatorname{Tr}\left\{  \mathcal{N}_{A\rightarrow
B}\left(  \tau_{A}\right)  \right\}  =\operatorname{Tr}\left\{  \tau
_{A}\right\}  $ for all input operators $\tau_{A}\in\mathcal{L}(\mathcal{H}%
_{A})$. It is trace non-increasing if $\operatorname{Tr}\left\{
\mathcal{N}_{A\rightarrow B}\left(  \tau_{A}\right)  \right\}  \leq
\operatorname{Tr}\left\{  \tau_{A}\right\}  $ for all $\tau_{A}\in
\mathcal{L}_{+}(\mathcal{H}_{A})$. A quantum channel is a linear map which is
completely positive and trace-preserving (CPTP). A positive operator-valued
measure (POVM) is a set $\left\{  \Lambda^{m}\right\}  $ of positive
semi-definite operators such that $\sum_{m}\Lambda^{m}=I$. For $X,Y\in
\mathcal{L}(\mathcal{H})$, let $\left\langle X,Y\right\rangle \equiv
\operatorname{Tr}\{X^{\dag}Y\}$ denote the Hilbert--Schmidt inner
product.\ The adjoint $\left(  \mathcal{M}_{A\rightarrow B}\right)  ^{\dag}%
$\ of a linear map $\mathcal{M}_{A\rightarrow B}$\ is the unique linear map
satisfying%
\begin{equation}
\left\langle Y_{B},\mathcal{M}_{A\rightarrow B}(X_{A})\right\rangle
=\langle\left(  \mathcal{M}_{A\rightarrow B}\right)  ^{\dag}(Y_{B}%
),X_{A}\rangle,
\end{equation}
for all $X_{A}\in\mathcal{L}(\mathcal{H}_{A})$ and $Y_{B}\in\mathcal{L}%
(\mathcal{H}_{B})$. A linear map $\mathcal{M}_{A\rightarrow B}$ is unital if
it preserves the identity, i.e., $\mathcal{M}_{A\rightarrow B}(I_{A})=I_{B}$.
It then follows that a linear map is unital if and only if its adjoint is
trace preserving. A linear map $\mathcal{M}_{A\rightarrow B}$\ is subunital if
$\mathcal{M}_{A\rightarrow B}(I_{A})\leq I_{B}$, and this is equivalent to the
adjoint of $\mathcal{M}_{A\rightarrow B}$ being trace non-increasing. A
quantum channel $\mathcal{U}:\mathcal{L}(\mathcal{H}_{A})\rightarrow
\mathcal{L}(\mathcal{H}_{B})$ is an isometric channel if it has the action
$\mathcal{U}(X_{A})=UX_{A}U^{\dag}$, where $X_{A}\in\mathcal{L}(\mathcal{H}%
_{A})$ and $U:\mathcal{H}_{A}\rightarrow\mathcal{H}_{B}$ is an isometry.

A quantum instrument\ is a quantum channel that accepts a quantum system as
input and outputs two systems:\ a classical one and a quantum one. More
formally, a quantum instrument is a collection $\{\mathcal{N}^{x}\}$\ of
completely positive trace non-increasing maps, such that the sum map $\sum
_{x}\mathcal{N}^{x}$ is a quantum channel. We can write the action of a
quantum instrument on an input operator $P$ as the following quantum channel:%
\begin{equation}
\label{eq:average-instr}P\rightarrow\sum_{x}\mathcal{N}^{x}(P)\otimes
|x\rangle\langle x|,
\end{equation}
where $\left\{  |x\rangle\right\}  $ is an orthonormal basis labeling the
classical output of the instrument.

The trace distance between two quantum states $\rho,\sigma\in\mathcal{D}%
(\mathcal{H})$\ is equal to $\left\Vert \rho-\sigma\right\Vert _{1}$. It has a
direct operational interpretation in terms of the distinguishability of these
states. That is, if $\rho$ or $\sigma$ are prepared with equal probability and
the task is to distinguish them via some quantum measurement, then the optimal
success probability in doing so is equal to $\left(  1+\left\Vert \rho
-\sigma\right\Vert _{1}/2\right)  /2$. The fidelity is defined as
$F(\rho,\sigma)\equiv\left\Vert \sqrt{\rho}\sqrt{\sigma}\right\Vert _{1}^{2}$
\cite{U73}, and more generally we can use the same formula to define $F(P,Q)$
if $P,Q\in\mathcal{L}_{+}(\mathcal{H})$. Uhlmann's theorem states that
\cite{U73}%
\begin{equation}
F(\rho_{A},\sigma_{A})=\max_{U}\left\vert \langle\phi^{\sigma}|_{RA}%
U_{R}\otimes I_{A}|\phi^{\rho}\rangle_{RA}\right\vert ^{2},
\label{eq:uhlmann-thm}%
\end{equation}
where $|\phi^{\rho}\rangle_{RA}$ and $|\phi^{\sigma}\rangle_{RA}$ are
purifications of $\rho_{A}$ and $\sigma_{A}$, respectively, and the
optimization is with respect to all isometries $U_{R}$. The same statement
holds more generally for $P,Q\in\mathcal{L}_{+}(\mathcal{H})$. We will also
use the notation $\sqrt{F}(\rho,\sigma)\equiv\left\Vert \sqrt{\rho}%
\sqrt{\sigma}\right\Vert _{1}$ to denote the ``root fidelity'' when
convenient. The direct-sum property of the fidelity is that%
\begin{equation}
\label{eq:direct-sum-fid}\sqrt{F}(\omega_{XS},\tau_{XS})=\sum_{x}\sqrt
{p_{X}(x)q_{X}(x)}\sqrt{F}(\omega_{S}^{x},\tau_{S}^{x}),
\end{equation}
for classical--quantum states%
\begin{align}
\omega_{XS}  &  \equiv\sum_{x}p_{X}(x)|x\rangle\langle x|_{X}\otimes\omega
_{S}^{x},\\
\tau_{XS}  &  \equiv\sum_{x}q_{X}(x)|x\rangle\langle x|_{X}\otimes\tau_{S}%
^{x}.
\end{align}

The quantum relative entropy $D(P\Vert Q)$ between $P,Q\in\mathcal{L}%
_{+}(\mathcal{H})$, with $P\neq0$, is defined as \cite{U62}
\begin{equation}
D(P\Vert Q)=\operatorname{Tr}\{P\left[  \log P-\log Q\right]  \}
\end{equation}
if $\operatorname{supp}(P)\subseteq\operatorname{supp}(Q)$ and as $+\infty$
otherwise. The relative entropy $D(P\Vert Q)$\ is non-negative if Tr$\left\{
P\right\}  \geq$ Tr$\left\{  Q\right\}  $, a result known as Klein's
inequality \cite{LR68}. Thus, for density operators $\rho$ and $\sigma$, the
relative entropy is non-negative, and furthermore, it is equal to zero if and
only if $\rho=\sigma$. The quantum relative entropy obeys the following
property:%
\begin{equation}
D(P\Vert Q)\geq D(P\Vert Q^{\prime}), \label{eq:rel-ent-dom}%
\end{equation}
for $P,Q,Q^{\prime}\in\mathcal{L}_{+}(\mathcal{H})$ such that $Q\leq
Q^{\prime}$. The following relationship between fidelity and quantum relative
entropy is well known (see, e.g., \cite{MDSFT13}):%
\begin{equation}
D(P\Vert Q)\geq-\log F(P,Q). \label{eq:rel-ent-fid}%
\end{equation}

The quantum entropy $H(\rho)$\ of a density operator $\rho$ is $H(\rho
)=-\operatorname{Tr}\{\rho\log\rho\}$. We often write this as $H(A)_{\rho}$ if
$\rho_{A}$ is the density operator for system $A$. The conditional entropy of
a bipartite density operator $\rho_{AB}$ is equal to $H(A|B)_{\rho}\equiv
H(AB)_{\rho}-H(B)_{\rho}$. The mutual information is equal to $I(A;B)_{\rho
}=H(A)_{\rho}-H(A|B)_{\rho}$. The conditional mutual information of a
tripartite state $\rho_{ABC}$ is equal to $I(A;B|C)_{\rho}=H(B|C)_{\rho
}-H(B|AC)_{\rho}$. The following identities are well known (see, e.g.,
\cite{W15book}):%
\begin{align}
H(A)_{\rho}  &  =-D(\rho_{A}\Vert I_{A}),\\
H(A|B)_{\rho}  &  =-D(\rho_{AB}\Vert I_{A}\otimes\rho_{B}%
),\label{eq:rel-ent-cond-ent}\\
I(A;B)_{\rho}  &  =D(\rho_{AB}\Vert\rho_{A}\otimes\rho_{B}).
\label{eq:mut-info-rel-ent}%
\end{align}

The following \textquotedblleft recoverability theorem\textquotedblright\ is
an enhancement of the monotonicity of quantum relative entropy (mentioned in
\eqref{eq:mono-enhance}) and was proved recently in \cite{Junge15}, by an
extension of the methods from \cite{W15}:%
\begin{equation}
D(\rho\Vert\sigma)\geq D(\mathcal{N}(\rho)\Vert\mathcal{N}(\sigma))-\log
F(\rho,(\mathcal{R}\circ\mathcal{N})(\rho)), \label{eq:recovery-fid}%
\end{equation}
where $\rho\in\mathcal{D}(\mathcal{H})$, $\sigma\in\mathcal{L}_{+}%
(\mathcal{H})$, $\mathcal{N}:\mathcal{L}(\mathcal{H})\rightarrow
\mathcal{L}(\mathcal{H}^{\prime})$ is a quantum channel, and $\mathcal{R}$ is
a recovery quantum channel of the following form:%
\begin{multline}
\mathcal{R}(Q)\equiv\operatorname{Tr}\{(I-\Pi_{\mathcal{N}(\sigma)}%
)Q\}\tau\label{eq:recovery-map}\\
+\int_{-\infty}^{\infty}dt\ p(t)\ \mathcal{R}_{\sigma,\mathcal{N}}^{t/2}(Q),
\end{multline}
where $\Pi_{\mathcal{N}(\sigma)}$ is the projection onto the support of
$\mathcal{N}(\sigma)$,
\begin{equation}
\tau\in\mathcal{D}(\mathcal{H}),\ \ \ \ \ \ \ \ p(t)\equiv\frac{\pi}{2}\left[
\cosh(t)+1\right]  ^{-1} \label{eq:pt-dist}%
\end{equation}
is a probability distribution on $t\in\mathcal{R}$,
\begin{equation}
\mathcal{U}_{\omega,t}(X)\equiv\omega^{it}X\omega^{-it}%
\end{equation}
for $\omega$ positive semi-definite,
\begin{equation}
\mathcal{P}_{\sigma,\mathcal{N}}(Q)\equiv\sigma^{1/2}\mathcal{N}^{\dag
}\!\left(  \mathcal{N}(\sigma)^{-1/2}Q\mathcal{N}(\sigma)^{-1/2}\right)
\sigma^{1/2}%
\end{equation}
is a completely positive, trace non-increasing map known as the Petz recovery
map \cite{Pet86,Pet88}, and $\mathcal{R}_{\sigma,\mathcal{N}}^{t}$ is a
rotated or \textquotedblleft swiveled\textquotedblright\ Petz recovery map,
defined as%
\begin{equation}
\mathcal{R}_{\sigma,\mathcal{N}}^{t}\equiv\mathcal{U}_{\sigma,-t}%
\circ\mathcal{P}_{\sigma,\mathcal{N}}\circ\mathcal{U}_{\mathcal{N}(\sigma),t}.
\end{equation}
In fact, the following stronger statement holds \cite{Junge15}%
\begin{multline}
D(\rho\Vert\sigma)\geq D(\mathcal{N}(\rho)\Vert\mathcal{N}(\sigma
))\label{eq:stronger-recovery}\\
-\int_{-\infty}^{\infty}dt\ p(t)\log F(\rho,(\mathcal{R}_{\sigma,\mathcal{N}%
}^{t/2}\circ\mathcal{N})(\rho)),
\end{multline}
which will be useful for our purposes here. The inequality in
\eqref{eq:recovery-fid} implies the following one:%
\begin{equation}
I(A;B|C)_{\rho}\geq-\log F(\rho_{ABC},\mathcal{R}_{C\rightarrow AC}(\rho
_{BC})), \label{eq:CMI-recovery}%
\end{equation}
where $\mathcal{R}_{C\rightarrow AC}$ is defined from \eqref{eq:recovery-map},
by taking $\sigma=\rho_{AC}$ and $\mathcal{N}=\operatorname{Tr}_{A}$. This
follows from the definition we gave for $I(A;B|C)_{\rho}$, the equality in
\eqref{eq:rel-ent-cond-ent}, and the inequality in \eqref{eq:recovery-fid}.
Similarly, the following holds as well:%
\begin{equation}
I(A;B|C)_{\rho}\geq-\int_{-\infty}^{\infty}dt\ p(t)\log F(\rho_{ABC}%
,\mathcal{R}_{\rho_{AC},\operatorname{Tr}_{A}}^{t/2}(\rho_{BC})),
\label{eq:CMI-recovery-stronger}%
\end{equation}
by taking $\sigma=\rho_{AC}$ and $\mathcal{N}=\operatorname{Tr}_{A}$.
Explicitly, the action of the recovery map $\mathcal{R}_{\rho_{AC}%
,\operatorname{Tr}_{A}}^{t/2}$ on an operator $\omega_{C}$ is given as
follows:%
\begin{equation}
\mathcal{R}_{\rho_{AC},\operatorname{Tr}_{A}}^{t/2}(\omega_{C})=\rho
_{AC}^{\frac{1-it}{2}}\left[  I_{A}\otimes\rho_{C}^{-\frac{1-it}{2}}\omega
_{C}\rho_{C}^{-\frac{1+it}{2}}\right]  \rho_{AC}^{\frac{1+it}{2}}.
\end{equation}

\section{Entropy gain}

It is well known that the quantum entropy cannot decrease under the action of
a subunital, positive, and trace-preserving map \cite{AU78,AU82}:%
\begin{equation}
H(\mathcal{N}(\rho))\geq H(\rho), \label{eq:ent-inc-unital}%
\end{equation}
where $\rho\in\mathcal{D}(\mathcal{H})$ and $\mathcal{N}:\mathcal{L}%
(\mathcal{H})\rightarrow\mathcal{L}(\mathcal{H}^{\prime})$ is a subunital,
positive, and trace-preserving map. This entropy inequality follows as a
simple consequence of the monotonicity of quantum relative entropy (now shown
to hold for positive, trace-preserving maps \cite{MR15}). That is,
\eqref{eq:ent-inc-unital} follows by picking $\sigma=I$ in
\eqref{eq:rel-ent-mono} and applying \eqref{eq:rel-ent-dom} and that
$\mathcal{N}$ is subunital, whereby%
\begin{align}
-H(\rho)  &  =D(\rho\Vert I)\\
&  \geq D(\mathcal{N}(\rho)\Vert\mathcal{N}(I))\\
&  \geq D(\mathcal{N}(\rho)\Vert I)\\
&  =-H(\mathcal{N}(\rho)).
\end{align}
This entropy inequality has a number of applications in quantum information
and other contexts.

The following theorem leads to an enhancement of \eqref{eq:ent-inc-unital}:

\begin{theorem}
\label{thm:entropy-gain}Let $\rho\in\mathcal{D}(\mathcal{H})$ and let
$\mathcal{N}:\mathcal{L}(\mathcal{H})\rightarrow\mathcal{L}(\mathcal{H}%
^{\prime})$ be a positive and trace-preserving map. Then%
\begin{equation}
H(\mathcal{N}(\rho))-H(\rho)\geq D(\rho\Vert(\mathcal{N}^{\dag}\circ
\mathcal{N})(\rho)).
\end{equation}

\end{theorem}

\begin{proof}
This follows because%
\begin{align}
&  H(\mathcal{N}(\rho))-H(\rho)\nonumber\\
&  =\operatorname{Tr}\{\rho\log\rho\}-\operatorname{Tr}\{\mathcal{N}(\rho
)\log\mathcal{N}(\rho)\}\\
&  =\operatorname{Tr}\{\rho\log\rho\}-\operatorname{Tr}\{\rho\mathcal{N}%
^{\dag}(\log\mathcal{N}(\rho))\}\\
&  \geq\operatorname{Tr}\{\rho\log\rho\}-\operatorname{Tr}\{\rho
\log(\mathcal{N}^{\dag}\circ\mathcal{N})(\rho)\}\\
&  =D(\rho\Vert(\mathcal{N}^{\dag}\circ\mathcal{N})(\rho)).
\end{align}
The second equality is from the definition of the adjoint. The inequality
follows from operator concavity of the logarithm and the operator Jensen
inequality for positive unital maps \cite{choi1974}\ (see also
\cite[Lemma~3.10]{Sutter15}).
\end{proof}

If $\mathcal{N}$ is additionally subunital~\footnote{Notice that a
trace-preserving positive map can be subunital only if the output space
dimension is not smaller than the input space dimension, i.e., $\dim
\mathcal{H}^{\prime}\ge\dim\mathcal{H}$. In particular, if $\dim
\mathcal{H}^{\prime}=\dim\mathcal{H}$, then subunitality is equivalent to
unitality.}, then Theorem~\ref{thm:entropy-gain} implies that $\mathcal{N}%
^{\dag}$ is trace non-increasing, which in turn implies that $D(\rho
\Vert(\mathcal{N}^{\dag}\circ\mathcal{N})(\rho))\geq0$ by Klein's inequality.
Thus, in this case, we obtain a significant strengthening of the well known
fact that the entropy increases under the action of a subunital, positive,
trace-preserving map.

The resulting entropy inequality also leads to an interpretation in terms of
recoverability, in the sense discussed in \cite{W15}. That is, we can take the
recovery map to be%
\begin{equation}
\mathcal{R}(Y)\equiv\mathcal{N}^{\dag}(Y)+\operatorname{Tr}%
\{(\operatorname{id}-\mathcal{N}^{\dag})(Y)\}\tau,
\end{equation}
where $\tau$ is any state in $\mathcal{D}(\mathcal{H})$, and we get that, for
all $\rho$,%
\begin{equation}
H(\mathcal{N}(\rho))-H(\rho)\geq D(\rho\Vert(\mathcal{R}\circ\mathcal{N}%
)(\rho)) \label{eq:subunital-rel-ent-bound}%
\end{equation}
by applying \eqref{eq:rel-ent-dom}, because $(\mathcal{R}\circ\mathcal{N}%
)(\rho)\geq(\mathcal{N}^{\dag}\circ\mathcal{N})(\rho)$. Note that
$\mathcal{R}$ is a positive map if $\mathcal{N}$ is. We also note that if
$\mathcal{N}$ is subunital the quantity $D\!\left(  \rho\Vert(\mathcal{R}%
\circ\mathcal{N})(\rho)\right)  $ can be viewed as a measure of how much
$\mathcal{N}$ deviates from being an isometry, being equal to zero if
$\mathcal{N}$ is an isometric channel and non-zero otherwise (here we could
also maximize the quantity with respect to input states $\rho$ and output
states~$\tau$).

Thus, what we find is an improvement over what we would get by applying
\eqref{eq:recovery-fid} or the main result of \cite{Sutter15}. First, there is
a mathematical advantage in the sense that $\mathcal{N}$ is not required to be
a channel, but it suffices for it to be a positive map. This addresses an open
question from \cite{MR15} for a very special case. Some might also consider
this to be a physical advantage as well, given that in some situations the
description of quantum dynamical evolutions is not given by a completely
positive map (see, e.g., \cite{B14} and references therein). Second, the
remainder term in \eqref{eq:subunital-rel-ent-bound} features the quantum
relative entropy and thus is stronger than the $-\log F$ bound in
\eqref{eq:recovery-fid} (cf. \eqref{eq:rel-ent-fid}) and the \textquotedblleft
measured relative entropy\textquotedblright\ term from \cite{Sutter15}.
Finally, note that Theorem~\ref{thm:entropy-gain} represents an improvement of
some of the results from \cite{Zhang20114163,C14}.

\subsection{Application to bosonic channels}

Theorem~\ref{thm:entropy-gain}\ finds application for practical bosonic
channels that have a long history in quantum information theory, in
particular, the pure-loss and quantum-limited amplifier channels, and even all
phase insensitive Gaussian channels \cite{WPGCRSL12}. A pure-loss channel is
defined from the following input-output Heisenberg-picture relation:%
\begin{equation}
\hat{b}=\sqrt{\eta}\hat{a}+\sqrt{1-\eta}\hat{e},
\end{equation}
where $\hat{a}$, $\hat{b}$, and $\hat{e}$ are the field-mode annihilation
operators representing the sender's input, the receiver's output, and the
environmental input of the channel. The parameter $\eta\in\left[  0,1\right]
$ represents the average fraction of photons that make it from the sender to
receiver. For the pure-loss channel, the environment is prepared in the vacuum
state. Let $\mathcal{B}_{\eta}$ denote the CPTP\ map corresponding to this
channel. A quantum-limited amplifier channel is defined from the following
input-output Heisenberg-picture relation:%
\begin{equation}
\hat{b}=\sqrt{G}\hat{a}+\sqrt{G-1}\hat{e}^{\dag},
\end{equation}
where $\hat{a}$, $\hat{b}$, and $\hat{e}$ have the same physical meaning as
given for the pure-loss channel. The parameter $G\in\lbrack1,\infty)$
represents the gain or amplification factor of the channel. For the
quantum-limited amplifier channel, the environment is prepared in the vacuum
state. Let $\mathcal{A}_{G}$ denote the CPTP\ map corresponding to this channel.

One of the critical insights of \cite{ISS11} is that these channels are
\textquotedblleft almost unital,\textquotedblright\ in the sense that%
\begin{equation}
\mathcal{B}_{\eta}(I)=\eta^{-1}I,\ \ \ \ \ \ \ \mathcal{A}_{G}(I)=G^{-1}I,
\end{equation}
and that their adjoints are strongly related, in the sense that%
\begin{align}
\mathcal{B}_{\eta}^{\dag}  &  =\eta^{-1}\mathcal{A}_{1/\eta},\\
\mathcal{A}_{G}^{\dag}  &  =G^{-1}\mathcal{B}_{1/G}.
\end{align}
Observe that the pure-loss channel is superunital and the amplifier channel is
subunital. These facts allow us to apply Theorem~\ref{thm:entropy-gain}\ and
the fact that $D(\rho\Vert c\sigma)=D(\rho\Vert\sigma)-\log c$ for $c>0$ to
find that%
\begin{align}
H(\mathcal{B}_{\eta}(\rho))-H(\rho)  &  \geq D(\rho\Vert(\mathcal{A}_{\frac
{1}{\eta}}\circ\mathcal{B}_{\eta})(\rho))+\log\eta,\\
H(\mathcal{A}_{G}(\rho))-H(\rho)  &  \geq D(\rho\Vert(\mathcal{B}_{\frac{1}%
{G}}\circ\mathcal{A}_{G})(\rho))+\log G.
\end{align}
These bounds demonstrate that a quantum-limited amplifier suffices as a
reversal channel for a pure-loss channel and vice versa. Note that the above
reversal is only good for weak losses and weak amplifiers (i.e., if
$\eta\approx1$ or $G\approx1$). We can also conclude that%
\begin{multline}
H((\mathcal{A}_{G}\circ\mathcal{B}_{\eta})(\rho))-H(\rho)\geq\\
D(\rho\Vert(\mathcal{A}_{1/\eta}\circ\mathcal{B}_{1/G}\circ\mathcal{A}%
_{G}\circ\mathcal{B}_{\eta})(\rho))+\log\left[  \eta G\right]  ,
\end{multline}
because%
\begin{equation}
(\mathcal{A}_{G}\circ\mathcal{B}_{\eta})^{\dag}=\left[  \eta G\right]
^{-1}\mathcal{A}_{1/\eta}\circ\mathcal{B}_{1/G}.
\end{equation}
The above bound applies to any phase insensitive quantum Gaussian channel,
given that any such channel can be written as a serial concatenation of a
pure-loss channel and a quantum-limited amplifier channel \cite{CGH06,PBLSC12}.

\subsection{Optimized entropy gain}

In \cite{A04}, the minimal entropy gain of a quantum channel was defined as%
\begin{equation}
G(\mathcal{N})\equiv\inf_{\rho}\left[  H(\mathcal{N}(\rho))-H(\rho)\right]  ,
\end{equation}
and the following bounds were established for a channel with the same input
and output Hilbert space $\mathcal{H}$:%
\begin{equation}
-\log\dim(\mathcal{H})\leq G(\mathcal{N})\leq0.
\end{equation}
(See also \cite{Holevo2010,Holevo2011,H11ISIT} for related work.) Applying
Theorem~\ref{thm:entropy-gain}\ gives the following alternate lower bound for
the entropy gain of a quantum channel:%
\begin{equation}
G(\mathcal{N})\geq\inf_{\rho}D(\rho\Vert(\mathcal{N}^{\dag}\circ
\mathcal{N})(\rho)).
\end{equation}

\subsection{Entropy gain in the presence of quantum side information}

A generalization of the entropy inequality in
\eqref{eq:subunital-rel-ent-bound}\ holds for the case of the conditional
quantum entropy, found by applying the same method:

\begin{corollary}
Let $\rho_{AB}\in\mathcal{D}(\mathcal{H}_{A}\otimes\mathcal{H}_{B})$ and
$\mathcal{N}_{A\to A^{\prime}}\otimes\operatorname{id}_{B}:\mathcal{L}%
(\mathcal{H}_{AB})\rightarrow\mathcal{L}(\mathcal{H}_{A^{\prime}B})$ be a
positive and trace-preserving map. Then%
\begin{multline}
H(A^{\prime}|B)_{\sigma}-H(A|B)_{\rho}\label{eq:cond-ent-unital}\\
\geq D(\rho_{AB}\Vert(\left(  \mathcal{N}_{A\rightarrow A^{\prime}}\right)
^{\dag}\circ\mathcal{N}_{A\rightarrow A^{\prime}})(\rho_{AB})),
\end{multline}
where $\sigma_{A^{\prime}B}\equiv(\mathcal{N}_{A\rightarrow A^{\prime}}%
\otimes\operatorname{id}_{B})(\rho_{AB})$.
\end{corollary}

\begin{proof}
This follows by applying Theorem~\ref{thm:entropy-gain} and definitions. From
Theorem~\ref{thm:entropy-gain}, we can conclude that%
\begin{multline}
H(A^{\prime}B)_{\sigma}-H(AB)_{\rho}\\
\geq D(\rho_{AB}\Vert(\left(  \mathcal{N}_{A\rightarrow A^{\prime}}\right)
^{\dag}\circ\mathcal{N}_{A\rightarrow A^{\prime}})(\rho_{AB})).
\end{multline}
Consider also that%
\begin{align}
&  H(A^{\prime}B)_{\sigma}-H(AB)_{\rho}\nonumber\\
&  =H(A^{\prime}B)_{\sigma}-H(B)_{\sigma}-\left[  H(AB)_{\rho}-H(B)_{\rho
}\right] \\
&  =H(A^{\prime}|B)_{\sigma}-H(A|B)_{\rho},
\end{align}
where we have used that $H(B)_{\rho}=H(B)_{\sigma}$. Combining these gives \eqref{eq:cond-ent-unital}.
\end{proof}

\begin{remark}
In the above corollary, note that we need not
necessarily take the map $\mathcal{N}_{A\rightarrow A^{\prime}}$ to be completely positive---we merely require that
$\mathcal{N}_{A\to A^{\prime}}\otimes\operatorname{id}_{B}$ be a positive map. For example, if system $B$ is a qubit, then we only require $\mathcal{N}_{A\rightarrow A^{\prime}}$ to be two-positive in order for the corollary to apply.
\end{remark}

\section{Information gain}

Groenewold originally defined the information gain of a quantum instrument
$\{\mathcal{N}^{x}\}$, when performed on a quantum state $\rho_{A}$, as
follows \cite{G71}:%
\begin{equation}
I_{G}(\{\mathcal{N}^{x} \},\rho_{A})\equiv H(\rho_{A})-\sum_{x}p_{X}%
(x)H(\rho_{A^{\prime}}^{x}), \label{eq:entropy-reduction}%
\end{equation}
where%
\begin{equation}
\rho_{A^{\prime}}^{x}\equiv\frac{\mathcal{N}_{A\rightarrow A^{\prime}}%
^{x}(\rho_{A})}{p_{X}(x)},\ \ \ \ \ \ \ p_{X}(x)\equiv\operatorname{Tr}%
\{\mathcal{N}_{A\rightarrow A^{\prime}}^{x}(\rho_{A})\}.
\label{eq:post-measurement-states}%
\end{equation}
This definition was based on the physical intuition that information gain
should be identified with the entropy reduction of the measurement. However,
it was later realized that the entropy reduction can be negative, and that
this happens if and only if the instrument is not an efficient measurement (an
efficient measurement is such that each $\mathcal{N}^{x}$ consists of a single
Kraus operator~\cite{Ozawa86,FJ01}).

Apparently without realizing the connection to Groenewold's information gain
of a measurement, Winter considered the operational, information-theoretic
task \cite{Win04}\ of determining the rate at which classical information
would need to be communicated from a sender to a receiver in order to simulate
the action of the measurement on a given state (if shared randomness is
allowed for free between sender and receiver). He called this task
\textquotedblleft measurement compression,\textquotedblright\ given that the
goal is to send the classical output of the measurement at the smallest rate
possible, in such a way that a third party would not be able to distinguish
the output of the protocol performed on many copies of $\phi_{RA}^{\rho}%
\,$from the same number of copies of the following state:%
\begin{equation}
\sigma_{RX}\equiv\sum_{x}\operatorname{Tr}_{A^{\prime}}\{(\operatorname{id}%
_{R}\otimes\mathcal{N}_{A\rightarrow A^{\prime}}^{x})(\phi_{RA}^{\rho
})\}\otimes|x\rangle\langle x|_{X},
\end{equation}
where $\phi_{RA}^{\rho}$ is a purification of $\rho$ and $\{|x\rangle\}$ is an
orthonormal basis for the classical output $X$ of the measurement. He found
that the optimal rate of measurement compression is equal to the mutual
information of the measurement $I(R;X)_{\sigma}$.

After Winter's development, Ref.~\cite{BHH08}\ suggested that the information
gain of the measurement should be defined as its mutual information. The
advantage of such an approach is that the mutual information $I(R;X)_{\sigma}$
is non-negative and has a clear operational interpretation. Furthermore, it is
equal to the entropy reduction in \eqref{eq:entropy-reduction} for efficient
measurements \cite{BHH08} and thus connects with Groenewold's original intuition.

Winter's result was later extended in two different directions. First,
Ref.~\cite{WHBH12}\ allowed for a correlated initial state $\rho_{AB}$, shared
between the sender and receiver before communication begins. In this case, the
optimal rate at which the sender needs to transmit classical information in
order to simulate the measurement is equal to the conditional mutual
information $I(R;X|B)_{\omega}$, where the conditional mutual information is
with respect to the following state:%
\begin{equation}
\omega_{RBX}\equiv\sum_{x}\operatorname{Tr}_{A^{\prime}}\{(\operatorname{id}%
_{R}\otimes\mathcal{N}_{A\rightarrow A^{\prime}}^{x})(\phi_{RAB}^{\rho
})\}\otimes|x\rangle\langle x|_{X},
\end{equation}
and $\phi_{RAB}^{\rho}$ is a purification of $\rho_{AB}$. We can thus call
$I(R;X|B)_{\omega}$ the information gain in the presence of quantum side
information (IG-QSI), and the information-processing task is known as
measurement compression with quantum side information \cite{WHBH12}. In
general, the IG-QSI is smaller than $I(RB;X)_{\omega}$, which is the rate at
which classical communication would need to be transmitted if the receiver
does not make use of the $B$ system. The other extension of Winter's result
was to determine the rate required to simulate the instrument on an arbitrary
input state, and the optimal rate was proved to be equal to the optimized
information gain%
\begin{equation}
\max_{\rho}I(R;X)_{\sigma},
\end{equation}
where the optimization is with respect to all input states $\rho_{A}$ leading
to a purification $\phi_{RA}^{\rho}$ \cite{BRW14}.

\subsection{General bounds on the information gain}

\label{sec:bound-inter}

Let us consider now the channel $\mathcal{N}$ associated to a quantum
instrument $\{\mathcal{N}^{x} \}$, as defined in (\ref{eq:average-instr}). By
defining the state $\sigma_{A^{\prime}X}$ as
\begin{align}
\sigma_{A^{\prime}X}  &  \equiv\mathcal{N}_{A\to A^{\prime}X}(\rho_{A})\\
&  =\sum_{x}\mathcal{N}^{x}_{A\to A^{\prime}}(\rho_{A})\otimes|x\rangle\langle x|_{X}\\
&  \equiv\sum_{x}p_{X}(x)\rho^{x}_{A^{\prime}}\otimes|x\rangle\langle x|_{X},
\end{align}
Theorem~\ref{thm:entropy-gain} in this case gives
\begin{align}
&  H(A^{\prime}X)_{\sigma}-H(A)_{\rho}\nonumber\\
&  =H(X)_{\sigma}+\sum_{x}p_{X}(x)H(\rho^{x}_{A^{\prime}})-H(\rho_{A})\\
&  =H(X)_{\sigma}-I_{G}(\{\mathcal{N}^{x} \},\rho_{A})\\
&  \ge D(\rho_{A}\|(\mathcal{N}^{\dag}\circ\mathcal{N})(\rho_{A})),
\end{align}
namely,
\begin{equation}
I_{G}(\{\mathcal{N}^{x} \},\rho_{A})\le H(X)_{\sigma}-D(\rho_{A}%
\|(\mathcal{N}^{\dag}\circ\mathcal{N})(\rho_{A})).
\end{equation}
The above upper bound on Groenewold's information gain is valid for any
quantum instrument $\{\mathcal{N}^{x} \}$ and any state $\rho$.

A much tighter bound can be given if the instrument is efficient. In this
case, it is easy to prove that the channel $\mathcal{N}$ defined in
(\ref{eq:average-instr}) is always subunital. This a consequence of the fact
that, if $\mathcal{N}^{x}_{A\to A^{\prime}}(C_{A})=V^{x}_{A\to A^{\prime}%
}\sqrt{\Lambda^{x}_{A}}C_{A}\sqrt{\Lambda^{x}_{A}}(V^{x}_{A\to A^{\prime}%
})^{\dag}$, where $\Lambda^{x}_{A} $ are POVM elements and $V^{x}_{A\to
A^{\prime}}$ are isometries, then $\mathcal{N}^{X}_{A\to A^{\prime}%
}(\openone_{A})=V^{x}_{A\to A^{\prime}}\Lambda^{x}_{A}(V^{x}_{A\to A^{\prime}%
})^{\dag}\le\openone_{A^{\prime}}$, for all $x$. Moreover, for efficient
measurements Groenewold's information gain and the mutual information of the
measurement $I(R;X)_{\sigma}$ are equal~\cite{BHH08}.

Thus, for efficient quantum instruments, Theorem~\ref{thm:entropy-gain} leads
to the following bound:
\begin{align}
H(X)_{\sigma}-I(R;X)_{\sigma}  &  =H(X|R)_{\sigma}\label{eq:meas-ent-diff}\\
&  \ge D(\rho_{A}\|(\mathcal{R}\circ\mathcal{N})(\rho_{A})).
\end{align}
where $\mathcal{R}$ is a recovery channel independent of $\rho$. In other
words, whenever the reference $R$ and the classical outcome $X$ are almost
perfectly correlated, i.e., $H(X|R)_{\sigma}\approx0$, then the action of the
instrument on the input state $\rho$ can be almost perfectly corrected on
average, i.e., $D(\rho_{A}\|(\mathcal{R}\circ\mathcal{N})(\rho_{A}))\approx0$.

We notice here that the quantity in~(\ref{eq:meas-ent-diff}) has been given an
interesting thermodynamical interpretation in Ref.~\cite{Jacobs09}, so that
the above bound can be seen as a strengthening of the second law for efficient
quantum measurements.

The above bound also provides a way to quantify, in an information-theoretic
way, ``how close'' a given POVM is to the ideal measurement of an observable:
one just need to prepare a state $\rho$ that commutes with that observable
(for example, the maximally mixed state $\openone/d$) and feed it through an
efficient measurement of the given POVM. The entropy difference in
(\ref{eq:meas-ent-diff}) is then a good indicator of such a ``closeness,''
being null whenever the POVM corresponds to a sharp measurement along the
diagonalizing basis. This method is somewhat similar to the approach
introduced in Ref.~\cite{BHOW14}.

\subsection{Information gain without quantum side information}

In what follows, we demonstrate how the refined entropy inequalities in
\eqref{eq:stronger-recovery} and \eqref{eq:CMI-recovery-stronger} have
implications for the information gain of a quantum measurement, both without
and with quantum side information. We begin with the simpler case of
information gain without quantum side information, a scenario considered in
\cite{BHH08}. The theorem below gives a lower bound on the information gain in
terms of how well one can recover from the action of an efficient measurement.
It can be viewed as a corollary of the more general statement given in
Theorem~\ref{thm:info-gain-QSI} in the next section.

\begin{theorem}
\label{thm:info-gain-no-QSI}Let $\rho\in\mathcal{D}(\mathcal{H}_{A})$ and
$\{\mathcal{N}^{x}\}$ be a quantum instrument, where each $\mathcal{N}%
^{x}:\mathcal{L}(\mathcal{H}_{A})\rightarrow\mathcal{L}(\mathcal{H}%
_{A^{\prime}})$. Then the following inequality holds%
\begin{equation}
I(R;X)_{\sigma}\geq-\log F(\sigma_{RX},\sigma_{R}\otimes\sigma_{X}).
\label{eq:meas-mut-info-fid}%
\end{equation}
If the quantum instrument is efficient, then the above inequality implies that%
\begin{equation}
I(R;X)_{\sigma}\geq-2\log\left[  \sum_{x}p_{X}(x)\sqrt{F}(\mathcal{U}%
_{A^{\prime}\rightarrow A}^{x}(\phi_{RA^{\prime}}^{\rho_{x}}),\phi_{RA}^{\rho
})\right]  , \label{eq:meas-mut-info-fid-1}%
\end{equation}
for some collection $\{\mathcal{U}_{A^{\prime}\rightarrow A}^{x}\}$, where
each $\mathcal{U}_{A^{\prime}\rightarrow A}^{x}$ is an isometric quantum
channel, $\phi_{RA^{\prime}}^{\rho_{x}}$ is a purification of $\rho
_{A^{\prime}}^{x}$ defined in \eqref{eq:post-measurement-states}, and
$p_{X}(x)$ is defined in \eqref{eq:post-measurement-states}.
\end{theorem}

\begin{proof}
The inequality in \eqref{eq:meas-mut-info-fid} is a simple consequence of
\eqref{eq:mut-info-rel-ent} and \eqref{eq:rel-ent-fid}. The inequality in
\eqref{eq:meas-mut-info-fid-1}\ follows because%
\begin{equation}
\sqrt{F}(\sigma_{RX},\sigma_{R}\otimes\sigma_{X})=\sum_{x}p_{X}(x)\sqrt
{F}(\phi_{R}^{\rho_{x}},\phi_{R}^{\rho}).
\end{equation}
Applying Uhlmann's theorem (see \eqref{eq:uhlmann-thm}), we can conclude that
there exist isometric channels $\mathcal{U}_{A^{\prime}\rightarrow A}^{x}$
such that $F(\phi_{R}^{\rho_{x}},\phi_{R}^{\rho})=F(\mathcal{U}_{A^{\prime
}\rightarrow A}^{x}(\phi_{RA^{\prime}}^{\rho_{x}}),\phi_{RA}^{\rho})$ for all
$x$.
\end{proof}

The implication of the inequality in \eqref{eq:meas-mut-info-fid-1} is that if
the information gain of the measurement is small, so that%
\begin{equation}
I(R;X)_{\sigma}\approx0, \label{eq:small-info-gain}%
\end{equation}
then it is possible to reverse the action of the measurement approximately, in
such a way as to restore the post-measurement state to the original state with
a fidelity%
\begin{equation}
\sum_{x}p_{X}(x)\sqrt{F}(\mathcal{U}_{A^{\prime}\rightarrow A}^{x}%
(\phi_{RA^{\prime}}^{\rho_{x}}),\phi_{RA}^{\rho})\approx1.
\end{equation}
We can thus view this result as a one-sided information-disturbance trade-off.
Note that \cite[Theorem~1]{BHH08} contains an observation related to this. The
observation above is also related to the general one from \cite{qip2002schu},
but the result above is stronger:\ an inability to find correction isometries,
which leads to a small fidelity, is a witness to having a large information
gain $I(R;X)_{\sigma}$, due to the presence of the negative logarithm in \eqref{eq:meas-mut-info-fid-1}.

The inequality in \eqref{eq:meas-mut-info-fid-1} also has an operational
implication for Winter's measurement compression task. If the information gain
is small, so that \eqref{eq:small-info-gain} holds, then the sender and
receiver can simulate the measurement with a high fidelity per copy of the
source state, in such a way that the sender does not need to transmit any
classical information at all. The receiver can just prepare many copies of
$\rho_{A}$ locally, perform the measurements, and deliver the outputs of the
measurements as the classical data. This situation occurs because the
reference system $R$ is approximately decoupled from the classical output, in
the sense that $F(\sigma_{RX},\sigma_{R}\otimes\sigma_{X})\approx1$ if
$I(R;X)_{\sigma}\approx0$.

\subsection{Information gain with quantum side information}

We can readily extend the above results to the case of quantum side
information, by employing the inequality in \eqref{eq:CMI-recovery}. This
leads to the following theorem:

\begin{theorem}
\label{thm:info-gain-QSI}Let $\rho_{AB}\in\mathcal{D}(\mathcal{H}_{A}%
\otimes\mathcal{H}_{B})$ and $\{\mathcal{N}^{x}\}$ be a quantum instrument,
where each $\mathcal{N}^{x}:\mathcal{L}(\mathcal{H}_{A})\rightarrow
\mathcal{L}(\mathcal{H}_{A^{\prime}})$. Then the following inequality holds%
\begin{multline}
I(R;X|B)_{\omega}\geq\label{eq:info-disturb-QSI}\\
-2\int_{-\infty}^{\infty}dt\ p(t)\log\left[  \sum_{x}p_{X}(x)\sqrt{F}%
(\omega_{RB}^{x},\mathcal{R}_{B}^{x,t/2}(\omega_{RB}))\right]  ,
\end{multline}
where%
\begin{equation}
\omega_{RBX}\equiv\sum_{x}\operatorname{Tr}_{A^{\prime}}\{\mathcal{N}%
_{A\rightarrow A^{\prime}}^{x}(\phi_{RAB}^{\rho})\}\otimes|x\rangle\langle
x|_{X},
\end{equation}
$\phi_{RAB}^{\rho}$ is a purification of $\rho_{AB}$,%
\begin{align}
\omega_{RBA^{\prime}}^{x}  &  \equiv\frac{\mathcal{N}_{A\rightarrow A^{\prime
}}^{x}(\phi_{RAB}^{\rho})}{p_{X}(x)},\\
p_{X}(x)  &  \equiv\operatorname{Tr}\{\mathcal{N}_{A\rightarrow A^{\prime}%
}^{x}(\phi_{RAB}^{\rho})\},
\end{align}
$\{p_{X}(x)\mathcal{R}_{B}^{x,t/2}\}$ is a quantum instrument defined by%
\begin{equation}
\mathcal{R}_{B}^{x,t/2}(\omega_{RB})\equiv\left(  \omega_{B}^{x}\right)
^{\frac{1-it}{2}}\omega_{B}^{-\frac{1-it}{2}}(\omega_{RB})\omega_{B}%
^{-\frac{1+it}{2}}\left(  \omega_{B}^{x}\right)  ^{\frac{1+it}{2}},
\label{eq:instrument-maps}%
\end{equation}
and $p(t)$ is defined in \eqref{eq:pt-dist}. If the instrument $\{\mathcal{N}%
^{x}\}$ is efficient, then the following inequality holds as well:%
\begin{equation}
I(R;X|B)_{\omega}\geq-2\int_{-\infty}^{\infty}dt\ p(t)\log\left[  \sum
_{x}p_{X}(x)\sqrt{F_{x,t}}\right]  , \label{eq:meas-QSI-correct}%
\end{equation}
for some collection $\{\mathcal{U}_{A\rightarrow A^{\prime}}^{x,t}\}$, where%
\begin{equation}
F_{x,t}\equiv F(\omega_{RBA^{\prime}}^{x},(\mathcal{R}_{B}^{x,t/2}%
\otimes\mathcal{U}_{A\rightarrow A^{\prime}}^{x,t})(\phi_{RBA}^{\rho}))
\end{equation}
and each $\mathcal{U}_{A\rightarrow A^{\prime}}^{x,t}$ is an isometric quantum channel.
\end{theorem}

\begin{proof}
We begin by proving the inequality in \eqref{eq:info-disturb-QSI}. Consider
that%
\begin{multline}
I(R;X|B)_{\omega}\geq\\
-\int_{-\infty}^{\infty}dt\ p(t)\log F(\omega_{RBX},\mathcal{R}_{B\rightarrow
BX}^{t/2}(\omega_{RB})),
\end{multline}
where%
\begin{equation}
\mathcal{R}_{B\rightarrow BX}^{t/2}(\omega_{RB})=\omega_{BX}^{\frac{1-it}{2}%
}\omega_{B}^{-\frac{1-it}{2}}\omega_{RB}\omega_{B}^{-\frac{1+it}{2}}%
\omega_{BX}^{\frac{1+it}{2}},
\end{equation}
which is a direct consequence of \eqref{eq:CMI-recovery-stronger}. By a direct
calculation, we find that%
\begin{equation}
\mathcal{R}_{B\rightarrow BX}^{t/2}(\omega_{RB})=\sum_{x}p_{X}(x)|x\rangle
\langle x|_{X}\otimes\mathcal{R}_{B}^{x,t/2}(\omega_{RB}),
\end{equation}
with $\mathcal{R}_{B}^{x,t/2}$ defined in \eqref{eq:instrument-maps}. This
then leads to the inequality in \eqref{eq:info-disturb-QSI}, by applying the
direct sum property of fidelity. The inequality in \eqref{eq:meas-QSI-correct}
is an application of Uhlmann's theorem, after observing that the rank-one
operator $\mathcal{R}_{B}^{x,t/2}(\phi_{RBA}^{\rho})$ purifies $\mathcal{R}%
_{B}^{x,t/2}(\omega_{RB})$ and the rank-one operator $\omega_{RBA^{\prime}%
}^{x}$ purifies $\omega_{RB}^{x}$. The aforementioned operators are rank-one
if the measurement is efficient (which is what we assumed in the statement of
the theorem).
\end{proof}

The implications of Theorem~\ref{thm:info-gain-QSI}\ are similar to those of
Theorem~\ref{thm:info-gain-no-QSI}, except they apply to a setting in which
quantum side information is available. If the information gain of the
measurement is small, so that%
\begin{equation}
I(R;X|B)_{\omega}\approx0, \label{eq:small-IG-QSI}%
\end{equation}
then it is possible to reverse the action of the measurement approximately, in
such a way as to restore the post-measurement state of systems $RA^{\prime}$
to the original state on systems $RA$ with a fidelity larger than%
\begin{equation}
\int_{-\infty}^{\infty}dt\ p(t)\sum_{x}p_{X}(x)\sqrt{F_{x,t}}\approx1.
\end{equation}
This follows from the concavity of the fidelity. The reversal operation
consists of two steps. First, Bob performs the instrument $\{p_{X}%
(x)\mathcal{R}_{B}^{x,t/2}\}$. He then forwards the outcomes to Alice, who
performs a channel corresponding to the inverse of the isometric quantum
channel $\mathcal{U}_{A\rightarrow A^{\prime}}^{x,t}$. Then, the average
fidelity is high if the information gain is small. We can view this result as
a one-sided information-disturbance trade-off which extends the aforementioned
one without quantum side information.

The inequality in \eqref{eq:meas-QSI-correct} also has an operational
implication for measurement compression with quantum side information
\cite{WHBH12}. If the IG-QSI\ is small, so that \eqref{eq:small-IG-QSI} holds,
then the sender and receiver can simulate the measurement with a high fidelity
per copy of the source state, in such a way that the sender does not need to
transmit any classical information at all. The receiver can just perform the
instrument $\{p_{X}(x)\mathcal{R}_{B}^{x,t/2}\}$ with probability $p(t)$ on
the individual $B$ systems of many copies of $\rho_{AB}$ and deliver the
classical outputs of the measurements as the classical data. This situation
occurs because the $X$ system of $\omega_{RBX}$ is approximately recoverable
from $B$ alone, in the sense that $\int_{-\infty}^{\infty}dt\ p(t)\sum
_{x}p_{X}(x)\sqrt{F_{x,t}}\approx1$ if $I(R;X|B)_{\omega}\approx0$. This
latter result might have implications for quantum communication complexity
(cf.~\cite{Touchette14}).

\section{Entropic disturbance}

Ref.~\cite{BH08} (see in particular Section~5 therein) considered the
possibility of introducing an entropic measure of average disturbance as
follows. Imagine that an initial ensemble of quantum states $\mathcal{E}%
=\{p_{X}(x);\rho_{A}^{x}\}_{x}$ is fed through a quantum channel
$\mathcal{N}:\mathcal{L}(\mathcal{H}_{A})\rightarrow\mathcal{L}(\mathcal{H}%
_{A^{\prime}})$. Consider the Holevo information of the initial ensemble
$\mathcal{E}$:
\begin{equation}
\chi(\mathcal{E})\equiv H(\rho_{A}^{\mathcal{E}})-\sum_{x}p_{X}(x)H(\rho
_{A}^{x}),
\end{equation}
where $\rho_{A}^{\mathcal{E}}$ denotes the average quantum state $\sum
_{x}p_{X}(x)\rho_{A}^{x}$. By the monotonicity of the Holevo information, the
following inequality holds
\begin{equation}
\Delta\chi(\mathcal{E})\equiv\chi(\mathcal{E})-\chi(\mathcal{N}(\mathcal{E}%
))\geq0,
\end{equation}
where by $\mathcal{N}(\mathcal{E})$ we mean the output ensemble $\{p_{X}%
(x);\mathcal{N}(\rho_{A^{\prime}}^{x})\}_{x}$.

It is known that the condition $\Delta\chi(\mathcal{E})=0$ implies the
existence of a recovery CPTP linear map $\mathcal{R}:\mathcal{L}%
(\mathcal{H}_{A^{\prime}})\to\mathcal{L}(\mathcal{H}_{A})$ such that
\begin{equation}
\mathcal{R}\circ\mathcal{N}(\rho^{x}_{A})=\rho^{x}_{A},
\end{equation}
for all $x$~\cite{HJPW04}. In Ref.~\cite{BH08} the question was considered,
whether a similar conclusion would hold also in the approximate case, but an
answer was given only in the case in which the input ensemble consists of two
mutually unbiased bases distributed with uniform prior, as done in~\cite{CW05}.

Recent results about approximate recoverability give a solution to this
problem, by demonstrating that there exists a recovery channel that can
approximately recover if the loss of Holevo information is small. In fact, a
special case of this problem was already solved in \cite[Corollary~16]{W15}
when $\mathcal{N}$ is a measurement channel. Here we establish the following
more general theorem:

\begin{theorem}
Let $\mathcal{E}=\{p_{X}(x),\rho^{x}_{A} \}$ be an ensemble of states in
$\mathcal{D}(\mathcal{H}_{A})$ and $\mathcal{N}:\mathcal{L}(\mathcal{H}%
_{A})\to\mathcal{L}(\mathcal{H}_{A^{\prime}})$ a quantum channel. Then there
exists a recovery channel $\mathcal{R}:\mathcal{L}(\mathcal{H}_{A^{\prime}%
})\to\mathcal{L}(\mathcal{H}_{A})$ such that
\begin{multline}
\chi(\mathcal{E})-\chi(\mathcal{N}(\mathcal{E}))\\
\ge-2\log\sum_{x} p_{X}(x)\sqrt{F}(\rho^{x}_{A},(\mathcal{R}\circ
\mathcal{N})(\rho^{x}_{A})). \label{eq:holevo-loss-recover}%
\end{multline}

\end{theorem}

\begin{proof}
Introduce an auxiliary system $X$ and the bipartite classical-quantum state
\begin{equation}
\rho_{XA}\equiv\sum_{x}p_{X}(x)|x\rangle \langle x|_{X}\otimes\rho_{A}^{x},
\end{equation}
where the vectors $\{|x\rangle\}$ are orthonormal in the Hilbert space
$\mathcal{H}_{X}$. In an analogous way, we also write
\begin{equation}
\sigma_{XA^{\prime}}\equiv\sum_{x}p_{X}(x)|x\rangle \langle x|_{X}\otimes\mathcal{N}%
_{A\rightarrow A^{\prime}}(\rho_{A}^{x}).
\end{equation}
Then,
\begin{align}
&  \chi(\mathcal{E})-\chi(\mathcal{N}(\mathcal{E}))\nonumber\\
&  =I(X;A)_{\rho}-I(X;A^{\prime})_{\sigma}\\
&  =-H(X|A)_{\rho}+H(X|A^{\prime})_{\sigma}\\
&  =D(\rho_{XA}\Vert I_{X}\otimes\rho_{A})\nonumber\\
&  \quad-D((\operatorname{id}_{X}\otimes\mathcal{N}_{A\rightarrow A^{\prime}%
})(\rho_{XA})\Vert(\operatorname{id}_{X}\otimes\mathcal{N}_{A\rightarrow
A^{\prime}})(I_{X}\otimes\rho_{A})).
\end{align}
We now invoke (\ref{eq:recovery-fid}), noticing that, in this case, the
operator $\sigma$ has the special form $I_{X}\otimes\rho_{A}$ and the noise
acts only locally, i.e., it has the form $\operatorname{id}_{X}\otimes
\mathcal{N}_{A\rightarrow A^{\prime}}$. These two facts together imply that
(\ref{eq:recovery-fid}) can be written in this case as follows:
\begin{multline}
\chi(\mathcal{E})-\chi(\mathcal{N}(\mathcal{E}))\\
\geq-\log F(\rho_{XA},(\operatorname{id}_{X}\otimes\mathcal{R}_{A^{\prime
}\rightarrow A})(\sigma_{XA^{\prime}})),
\end{multline}
where $\mathcal{R}:\mathcal{L}(\mathcal{H}_{A^{\prime}})\rightarrow
\mathcal{L}(\mathcal{H}_{A})$ is a suitable recovery channel. Finally, we make
use of the direct sum property of the fidelity in (\ref{eq:direct-sum-fid}) to
establish \eqref{eq:holevo-loss-recover}.
\end{proof}

\section{Completely positive trace-preserving maps and quantum data
processing}

\label{sec:CP-DP}This section demonstrates how the inequality in
\eqref{eq:CMI-recovery}\ and the Alicki--Fannes--Winter inequality
\cite{AF04,Winter15}\ lead to a robust version of the main conclusion of
\cite{B14}, which links the data processing inequality to complete positivity
of open quantum systems dynamics.

There, the problem of open quantum systems evolution in the presence of
initial system-environment correlations was considered. In fact, if the system
and its surrounding environment are correlated already before the interaction
governing their joint evolution is turned on, then in general there does not
necessarily exist a linear (let alone positive or even completely positive)
map describing the reduced dynamics of the system~\cite{BP02}. Ref.~\cite{B14}
proposed to use the quantum data-processing inequality as a criterion to
establish whether the system's reduced dynamics are compatible with a CPTP
linear map or not.

The operational framework considered in \cite{B14} can be summarized as follows:

\begin{enumerate}
\item It is assumed that possible joint system-environment states belong to a
known family of states that constitutes the promise to the problem. It is also
assumed that such a family is ``steerable,'' namely, that there exists a
tripartite density operator $\rho_{RQE}$ such that the reference $R$ is able
to steer all possible bipartite system-environment states in the family. Such
a condition encompasses essentially all cases considered in the literature. We
therefore assume that, at some initial time $t=\tau$, the system-environment
correlations can be described by means of \textit{one} given tripartite state
$\rho_{RQE}$.

\item Moving to the next instant in time, $t=\tau+\Delta$, the
system-environment pair has evolved according to some isometry $V:QE\to
Q^{\prime}E^{\prime}$, while the reference $R$ remains unchanged. The
tripartite configuration $\rho_{RQE}$ has correspondingly evolved to the
tripartite configuration $\sigma_{RQ^{\prime}E^{\prime}}=(I_{R}\otimes
V_{QE})\rho_{RQE}(I_{R}\otimes V_{QE}^{\dag})$.

\item Only at this point we focus on the reduced reference-system dynamics
(i.e., the transformation mapping $\rho_{RQ}$ to $\sigma_{RQ^{\prime}}$),
checking whether these are compatible with the application of a CPTP linear
map on the system $Q$ \emph{alone}. More explicitly, we check whether there
exists a CPTP linear map $\mathcal{E}:Q\to Q^{\prime}$ such that
$\sigma_{RQ^{\prime}}=(\operatorname{id}_{R}\otimes\mathcal{E}_{Q})(\rho
_{RQ})$.
\end{enumerate}

The following theorem generalizes to the approximate scenario the insight
provided in Ref.~\cite{B14}.

\begin{theorem}
Fix a tripartite configuration $\rho_{RQE}$. Suppose that the data processing
inequality holds approximately for all joint system-environment evolutions
$V_{QE\rightarrow Q^{\prime}E^{\prime}}$, i.e.,%
\begin{equation}
I(R;Q^{\prime})_{\sigma}\leq I(R;Q)_{\rho}+\varepsilon, \label{eq:approx-DP}%
\end{equation}
where $\varepsilon>0$ and%
\begin{equation}
\sigma_{RQ^{\prime}E^{\prime}}=V_{QE\rightarrow Q^{\prime}E^{\prime}}%
\rho_{RQE}V_{QE\rightarrow Q^{\prime}E^{\prime}}^{\dag}.
\end{equation}
Then the conditional mutual information is nearly equal to zero:%
\begin{equation}
I(R;E|Q)_{\rho}\leq\varepsilon, \label{eq:small-CMI}%
\end{equation}
and the reduced dynamics are approximately CPTP, i.e., to every unitary
interaction $V_{QE\rightarrow Q^{\prime}E^{\prime}}$, there exists a CPTP\ map
$\mathcal{E}_{Q\rightarrow Q^{\prime}}$ such that%
\begin{equation}
-\log F( \sigma_{RQ^{\prime}},\mathcal{E}_{Q\rightarrow Q^{\prime}}( \rho
_{RQ}) ) \leq\varepsilon. \label{eq:approx-CPTP}%
\end{equation}

\end{theorem}

\begin{proof}
We begin by proving (\ref{eq:small-CMI}) with the same approach used in
\cite{B14}. Consider the particular evolution in which $Q^{\prime}=QE$ and
system $E^{\prime}$ is trivial. The assumption that data processing holds
approximately gives that%
\begin{equation}
I(R;Q)_{\rho}+\varepsilon\geq I(R;Q^{\prime})_{\sigma}=I(R;QE)_{\rho}.
\end{equation}
We can rewrite this inequality using the chain rule for conditional mutual
information as%
\begin{equation}
\varepsilon\geq I(R;QE)_{\rho}-I(R;Q)_{\rho}=I(R;E|Q)_{\rho},
\label{eq:small-CMI-2}%
\end{equation}
which proves (\ref{eq:small-CMI}). Now, from the inequality in
\eqref{eq:CMI-recovery}, we know that there exists a recovery map
$\mathcal{R}_{Q\rightarrow QE}$\ such that%
\begin{equation}
I(R;E|Q)_{\rho}\geq-\log F\left(  \rho_{RQE},\mathcal{R}_{Q\rightarrow
QE}\left(  \rho_{RQ}\right)  \right)
\end{equation}
Since the fidelity is invariant with respect to unitaries, we find
(abbreviating $V_{QE\rightarrow Q^{\prime}E^{\prime}}$ as $V$) that
\begin{align}
&  F\left(  \rho_{RQE},\mathcal{R}_{Q\rightarrow QE}\left(  \rho_{RQ}\right)
\right) \nonumber\\
&  =F\left(  V\rho_{RQE}V^{\dag},V\mathcal{R}_{Q\rightarrow QE}\left(
\rho_{RQ}\right)  V^{\dag}\right) \\
&  =F\left(  \sigma_{RQ^{\prime}E^{\prime}},V\mathcal{R}_{Q\rightarrow
QE}\left(  \rho_{RQ}\right)  V^{\dag}\right) \\
&  \leq F\left(  \sigma_{RQ^{\prime}},\text{Tr}_{E^{\prime}}\left\{
V\mathcal{R}_{Q\rightarrow QE}\left(  \rho_{RQ}\right)  V^{\dag}\right\}
\right)  ,
\end{align}
where the inequality follows from monotonicity of fidelity under the
discarding of subsystems. By defining the channel%
\begin{equation}
\mathcal{E}_{Q\rightarrow Q^{\prime}}\left(  \cdot\right)  \equiv
\text{Tr}_{E^{\prime}}\left\{  V_{QE\rightarrow Q^{\prime}E^{\prime}%
}\mathcal{R}_{Q\rightarrow QE}\left(  \cdot\right)  V_{QE\rightarrow
Q^{\prime}E^{\prime}}^{\dag}\right\}  ,
\end{equation}
we find that%
\begin{equation}
\varepsilon\geq-\log F\left(  \sigma_{RQ^{\prime}},\mathcal{E}_{Q\rightarrow
Q^{\prime}}\left(  \rho_{RQ}\right)  \right)  ,
\end{equation}
establishing (\ref{eq:approx-CPTP}).
\end{proof}

The following theorem provides a converse.

\begin{theorem}
Suppose that the reduced dynamics are approximately CPTP, i.e., that to every
unitary interaction $V_{QE\rightarrow Q^{\prime}E^{\prime}}$ leading to%
\begin{equation}
\sigma_{RQ^{\prime}}=\operatorname{Tr}_{E^{\prime}}\left\{  V_{QE\rightarrow
Q^{\prime}E^{\prime}}\rho_{RQE}V_{QE\rightarrow Q^{\prime}E^{\prime}}^{\dag
}\right\}  ,
\end{equation}
there exists a CPTP\ map $\mathcal{E}_{Q\rightarrow Q^{\prime}}$\ such that%
\begin{equation}
\frac{1}{2}\left\Vert \sigma_{RQ^{\prime}}-\mathcal{E}_{Q\rightarrow
Q^{\prime}}\left(  \rho_{RQ}\right)  \right\Vert _{1}\leq\varepsilon,
\label{eq:approx-CPTP-1}%
\end{equation}
where $\varepsilon\in\left[  0,1\right]  $. Then the quantum data processing
inequality is satisfied approximately, in the sense that
\begin{multline}
I(R;Q^{\prime})_{\sigma}\leq I(R;Q)_{\rho}\\
+2\varepsilon\log\left\vert R\right\vert +\left(  1+\varepsilon\right)
h_{2}(\varepsilon/\left(  1+\varepsilon\right)  ),
\end{multline}
and the conditional mutual information is nearly equal to zero as well:%
\begin{equation}
I(R;E|Q)_{\rho}\leq2\varepsilon\log\left\vert R\right\vert +\left(
1+\varepsilon\right)  h_{2}(\varepsilon/\left(  1+\varepsilon\right)  ).
\end{equation}

\end{theorem}

\begin{proof}
This follows directly from the assumption in (\ref{eq:approx-CPTP-1}), the
Alicki--Fannes--Winter inequality, and the quantum data processing inequality:%
\begin{align}
&  I(R;Q^{\prime})_{\sigma}\nonumber\\
&  =H(R)_{\sigma}-H(R|Q^{\prime})_{\sigma}\\
&  =H(R)_{\mathcal{E}\left(  \rho\right)  }-H(R|Q^{\prime})_{\sigma}\\
&  \leq H(R)_{\mathcal{E}\left(  \rho\right)  }-H(R|Q^{\prime})_{\mathcal{E}%
\left(  \rho\right)  }\nonumber\\
&  \qquad+2\varepsilon\log\left\vert R\right\vert +\left(  1+\varepsilon
\right)  h_{2}(\varepsilon/\left(  1+\varepsilon\right)  )\\
&  =I(R;Q^{\prime})_{\mathcal{E}\left(  \rho\right)  }\nonumber\\
&  \qquad+2\varepsilon\log\left\vert R\right\vert +\left(  1+\varepsilon
\right)  h_{2}(\varepsilon/\left(  1+\varepsilon\right)  )\\
&  \leq I(R;Q)_{\rho}+2\varepsilon\log\left\vert R\right\vert +\left(
1+\varepsilon\right)  h_{2}(\varepsilon/\left(  1+\varepsilon\right)  ).
\end{align}
The inequality for conditional mutual information follows the same reasoning
we used to arrive at (\ref{eq:small-CMI-2}).
\end{proof}

\section{Conclusion}

\label{sec:conclusion}We have shown how recent results regarding
recoverability give enhancements to several entropy inequalities, having to do
with entropy gain, information gain, disturbance, and complete positivity of
open quantum systems dynamics. Our first result is a remainder term for the
entropy gain of a quantum channel, which for unital channels is stronger than
that which is obtained by directly applying the results of
\cite{Junge15,Sutter15}. This result implies that a small increase in entropy
under a subunital channel is a witness to the fact that the channel's adjoint
can be used as a recovery channel to undo the action of the original channel.
Our second result regards the information gain of a quantum measurement, both
without and with quantum side information. We find here that a small
information gain implies that it is possible to undo the action of the
original measurement (if it is efficient). The result also has operational
ramifications for the information-theoretic tasks known as measurement
compression without and with quantum side information. Our third result
provides an information-theoretic measure of disturbance, introduced
in~\cite{BH08}, a strong operational meaning. We finally provide a robust
extension of the main result of \cite{B14}, establishing that the reduced
dynamics of a system-environment interaction are approximately CPTP\ if and
only if the data processing inequality holds approximately.

\bigskip

\begin{acknowledgments}
We acknowledge helpful discussions with Mario Berta, Ryan Glasser, Eneet Kaur,
Prabha Mandayam, David Reeb, Maksim E.~Shirokov, David Sutter, Marco
Tomamichel, Dave Touchette, Andreas Winter, and Lin Zhang. FB acknowledges
support from the JSPS KAKENHI, No.~26247016. SD\ and MMW acknowledge the NSF
under Award No.~CCF-1350397 and startup funds from the Department of Physics
and Astronomy and the Center for Computation and Technology at Louisiana State University.
\end{acknowledgments}

\bibliographystyle{plain}
\bibliography{Ref}

\end{document}